\documentclass[10pt, twocolumn]{IEEEtran}
\usepackage{bbding}
\usepackage{url}
\usepackage{booktabs}
\usepackage{subfigure}
\usepackage{multirow}
\usepackage{fancybox}
\usepackage[dvipsnames,svgnames,x11names]{xcolor}
\usepackage{graphicx}
\usepackage{bbm}
\usepackage{cite}
\usepackage{epsfig, psfrag, amsmath, amssymb, amsfonts, latexsym, eucal, balance, hyperref, indentfirst,bookmark, amsthm}

\usepackage{soul}

\newtheorem{thm}{Theorem}
\newtheorem{cor}{Corollary}
\newtheorem{lem}{Lemma}

\newtheorem{defn}{Definition}
\newtheorem{rem}{Remark}

\makeatletter

\newcommand{\Rmnum}[1]{\expandafter\@slowromancap\romannumeral #1@}
\makeatother

\begin{document}
\title{Spatio-temporal Modeling for Massive and Sporadic Access}
\author{Yi Zhong, \emph{Member, IEEE}, Guoqiang Mao, \emph{Fellow, IEEE}, Xiaohu Ge, \emph{Senior Member, IEEE}, and Fu-chun Zheng, \emph{Senior Member, IEEE}
\thanks{Yi Zhong and Xiaohu Ge are with School of Electronic Information and Communications, Huazhong University of Science and Technology, Wuhan, P. R. China (e-mail: \{yzhong, xhge\}@hust.edu.cn). Guoqiang Mao is with the
School of Telecommunications Engineering, Xidian University (e-mail: g.mao@ieee.org). Fu-chun Zheng is with the School of Electronic and Information Engineering, Harbin Institute of Technology Shenzhen (e-mail: zhengfuchun@hit.edu.cn).

The corresponding author is Xiaohu Ge.

This research was supported by the National Natural Science Foundation of China (NSFC)
grant No. 61701183 and the Fundamental Research Funds for the Central Universities
through grant 2018KFYYXJJ139.
}
}
\maketitle
\begin{abstract}
The vision for smart city imperiously appeals to the implementation of Internet-of-Things (IoT), some features of which, such as massive access and bursty short packet transmissions, require new methods to enable the cellular system to seamlessly support its integration. Rigorous theoretical analysis is indispensable to obtain constructive insight for the networking design of massive access. In this paper, we propose and define the notion of massive and sporadic access (MSA) to quantitatively describe the massive access of IoT devices. We evaluate the temporal correlation of interference and successful transmission events, and verify that such correlation is negligible in the scenario of MSA. In view of this, in order to resolve the difficulty in any precise spatio-temporal analysis where complex interactions persist among the queues, we propose an approximation that all nodes are moving so fast that their locations are independent at different time slots. Furthermore, we compare the original static network and the equivalent network with high mobility to demonstrate the effectiveness of the proposed approximation approach. The proposed approach is promising for providing a convenient and general solution to evaluate and design the IoT network with massive and sporadic access.
\end{abstract}

%

\section{Introduction}
\subsection{Motivations}
Evolution of smart terminals (smart phones, smart watches, and intelligent glasses, etc.) has spawned a rich diversity of new applications, such as Virtual Reality, Augmented Reality, Industry 4.0 and so on, which post new challenges to wireless networks hosting these applications (see Figure \ref{fig:mmtc}).
A major branch of these applications is the access of massive Internet-of-Things (IoT) devices, also considered as one of the three main application scenarios of the fifth generation mobile communications system (5G) \cite{7736615, 7565189}.
The most promising applications of IoT include long-term environmental monitoring with limited energy consumption, smart city scenarios with millions of sensors, low-delay and high-reliability scenarios in wireless factory control etc.

\begin{figure}
\centering
\includegraphics[width=0.45\textwidth]{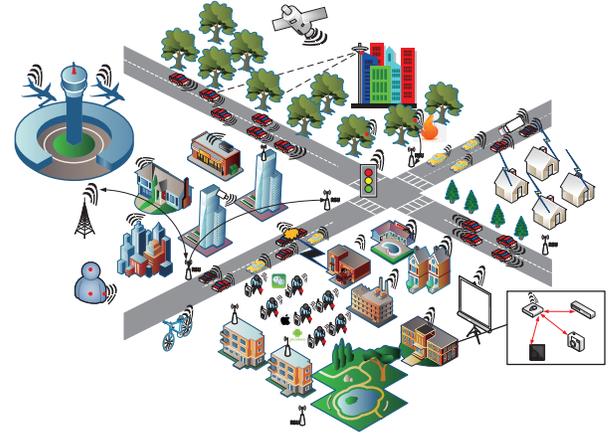}
\caption{The practical scenario for massive access of IoT devices.}
\label{fig:mmtc}
\end{figure}

For a typical IoT application, the number of wireless devices may reach 300,000 in each cell \cite{6736746}. In addition to the huge number of devices, another distinctive feature of the IoT applications is that the access requests from these massive number of devices are sporadic. The characteristic of sporadic access mainly manifests in the following aspects.
\begin{itemize}
\item
The users access the wireless network sporadically. For instance, the data generated by the IoT devices may be periodic and driven by some regularly occurring events. Typical application is the use of intelligent water meters, which report the measured data periodically.

\item
The size of each packet is small, and the data rate is low. The size of packets for a typical IoT application may go down to a few bytes. Meanwhile, the data rate for each user may be around 10 kb/s.
\end{itemize}

In view of this, we propose the notion of \emph{massive and sporadic access} (MSA), the precise mathematical description of which will be given in the following sections.
The massive and sporadic properties for such kind of accesses bring new challenges for the design of an efficient air interface. For the traditional wireless networks, a large overhead within all layers in the protocol stack is required to implement various network functions, such as access control, reliable transmission, authentication, security, and so on. However, for the MSA, the amount of information in each attempt of transmission may be only a few bytes. The large amount of overhead information generated by the protocol stack in the traditional wireless networks may greatly degrade the efficiency of MSA.
Moreover, the interference as well as other networking features of MSA may also be very different from the traditional wireless networks. Therefore, the scenarios of MSA should be modeled and analyzed qualitatively so that the constructive insight for the design of wireless networks with MSA could be obtained.

Though the point process theory has been widely used to evaluate the instantaneous performance metrics for a wireless network such as the coverage probability and the achievable rate \cite{haenggi2012stochastic, 6042301}, it is still not convenient for the characterization of many other significant performance metrics such as the delay when the queueing process is considered. In particular, when it is necessary to characterize the coupling between the traffic with spatio-temporal variation and the performance of a wireless network, rigorous analysis based on the point process theory becomes inapplicable \cite{7886285}.
The main difficulty of introducing the queueing analysis to the performance evaluation of a large wireless network, which is non-negligible for a practical system, lies in the complicated interaction among the queues \cite{rao1988stability, ephremides1987delay, telatar1995combining, 6691293}, i.e., the serving rates of all queues in a wireless network are highly coupled with the statuses of the queues (i.e., empty or not).

However, for the scenario of MSA, it is very unlikely that a node will be continuously active over consecutive time slots. Thus, the set of active nodes that cause interference changes dramatically over the time. For the case of ``extremely sporadic'', the sets of active nodes at different moments may not intersect since the probability that a node is active at two different moments is very small, in which case the interference might be considered as independent at different time slots. For the case of ``extremely massive'', an active node can appear anywhere in the plane. Intuitively, the analysis of the scenario of MSA becomes approximately equivalent to the analysis of a  network whose nodes move so fast that there is no coupling between the queues, thereby greatly reducing the analytical difficulties. Therefore, in this paper, we propose to evaluate the scenario of MSA by the equivalence of high mobility \cite{6380497, net:Stamatiou10cl} where the nodes move so fast that the location of a node in the subsequent time slots can be considered as totally independent from that in the current time slot.

\subsection{Related Works}
In order to characterize the massive and sporadic properties of MSA, we use the
combination of the point process theory and the queueing theory to model both the spatial distribution of the massive number of devices and the sporadic arrival of packets at each device. The stochastic geometry tools, especially the point process theory, have been widely used to model the spatial topology of wireless networks in recent years \cite{haenggi2012stochastic,haenggi2009stochastic}. For example, the Poisson point process (PPP) has been used to analyze the coverage probability, the achievable rate and other performance metrics in cellular networks \cite{baccelli2009stochastic2,6042301, ge2015spatial}.
Related works using the point process theory to evaluate the performance of IoT applications include \cite{78011591}, where two single-hop relaying schemes exclusively designed for the IoT are proposed and analyzed. By characterizing the received signal and interference powers using the point process theory, the authors in \cite{78011591} derived the outage probability and the maximum density of IoT devices that can be supported under an outage constraint. In \cite{8400530}, a framework to evaluate the end-to-end outage probability and the uplink data transmission rate in a single-hop relay network for IoT is proposed and evaluated using stochastic geometry. The authors in \cite{8354945} propose an analytical framework based on stochastic geometry to investigate the system performance in terms of the average success probability and the average number of simultaneously served IoT devices. The authors in \cite{7937902} also present a tractable analytical framework but to investigate the signal-to-interference ratio (SIR), thereby deriving the success probability, the average number of successful IoT devices and the probability of successful channel utilization for the cellular-based IoT network.
The authors in \cite{8554298} and \cite{8125754} evaluate the performance of massive non-orthogonal multiple access (NOMA) system.
The above-mentioned works focus on the performance metrics by considering a snapshot of the network. As for the analysis of the longer-term metrics such as the interference correlation and the delay, which requires the description of temporal variation, these approaches proposed in existing works are unable to provide an effective solution.
Moreover, the analysis based on a snapshot of a network cannot completely capture the sporadic property of MSA.

The works related to interference correlation in static Poisson networks include
\cite{net:Ganti09CL, net:Haenggi13twc, net:Zhong14twc}, where the interference correlation, as well as the ways to reduce its impact, is explored.
The analysis of temporal variation without considering the queueing process includes that for the local delay, which is defined as the number of time slots required for a packet to be successfully transmitted assuming that the networks are backlogged \cite{baccelli2010new, net:Haenggi13tit, net:Gong13twc}. However, since the queueing process is ignored in the evaluation of local delay, the obtained results reflect the practical situation accurately especially for those cases where random arrival of traffic has a great impact on the network performance.

Related works considering both the spatial distribution of nodes and the temporal variation of traffic could be found in \cite{blaszczyszyn2015performance, sapountzis2015analytical, abbas2015mobility}, in which the wireless traffic is modeled based on the granularity of total traffic in each cell. In\cite{7917340}, a traffic-aware spatio-temporal model is proposed for the Internet of Things (IoT) supported by the uplink of a cellular network. The stability for three different transmission strategies are evaluated.
In \cite{8408843}, the random access mechanism in the cellular-based massive IoT networks is evaluated based on a spatio-temporal model for the wireless traffic, where the spatial topology is modeled by using tools from  stochastic geometry and the evolutions of queues are assessed based on stochastic process.
In \cite{7842367}, a user-centric mobility management mechanism is proposed to cope with the spatial movements of users and the temporal correlation of wireless channels in ultra-dense networks. In order to model the flow at individual users, our previous work \cite{7886285} combines the point process theory and the queueing theory, and bounds the statistical distribution of the signal-to-interference ratio (SIR) and the delay in heterogeneous cellular networks. Along this line of thought, subsequent works such as \cite{8436053} and \cite{8335767} explored the delay and security performance in wireless networks.

\subsection{Contributions}
In this paper, we quantitatively model the spatio-temporal properties for the scenario of MSA. Afterwards, we explore the temporal correlation of interference and successful transmission events at different time slots. Then, we verify that these correlations are indeed negligible for the scenario of MSA.
In view of this, we propose to evaluate the scenario of MSA by the analysis of an equivalent network in which all nodes are moving extremely fast. In order to demonstrate the accuracy of the proposed approach, we further compare the performance of the original static MSA network and that of the equivalent network of high mobility. The main contributions are summarized as follows.

\begin{itemize}
\item We quantitatively characterize the spatial randomness in deployment and the temporal evolution of queues for \emph{massive and sporadic access (MSA)}, and discuss the interference-limited regime and the noise-limited regime for various configurations, which play an important role in practical MSA system design.
\item Based on the spatio-temporal characterization, we evaluate the temporal correlation between different time slots for interference and successful transmission, and demonstrate that such correlation is negligible for MSA. Hence, we propose an approximation approach of high mobility equivalence, which significantly reduces the analytical complexity and provides an elegant solution to evaluate MSA.
\item By the proposed approximation approach of high mobility equivalence, we derive the closed-form expressions for the success probability, the mean delay and the average queue length. Numerical and simulation results verify the accuracy of the proposed approximation approach.
\end{itemize}

The remaining parts of the paper are organized as follows. Section \ref{sec:model} describes the spatial distribution model and the arrival process. Section \ref{sec:mathMSA} gives the mathematical definition and evaluates the effect of parameters for the scenario of MSA. Section \ref{sec:equivalence} discusses the temporal correlation of interference and successful transmission events by considering a backlogged network. Section \ref{sec:accuracy} assesses the accuracy of using the high mobility equivalence to analyze the scenario of MSA. Finally, Section \ref{sec:conclusions} concludes the paper.

\section{System Model}
\label{sec:model}
Without loss of generality, we consider the commonly used \emph{Poisson bipolar model} (see \cite[Definition 5.8]{haenggi2012stochastic} and \cite{haenggi2015meta}) to characterize the spatial distribution of the devices. In such a model, the spatial distribution of the transmitters is modeled as a homogeneous PPP $\Phi=\{x_i\}\in \mathbb{R}^2$ with intensity $\lambda$, and each transmitter is associated with a receiver at a fixed distance $r_0$ and a random orientation (see Figure \ref{fig:sysmodel}). We consider a typical link with the receiver located at the origin and the transmitter located at $x_0$. Then, $|x_0|=r_0$ is the distance between the typical transmitter at $x_0$ and the typical receiver at the origin.

\begin{figure*}
\centering
\includegraphics[width=0.9\textwidth]{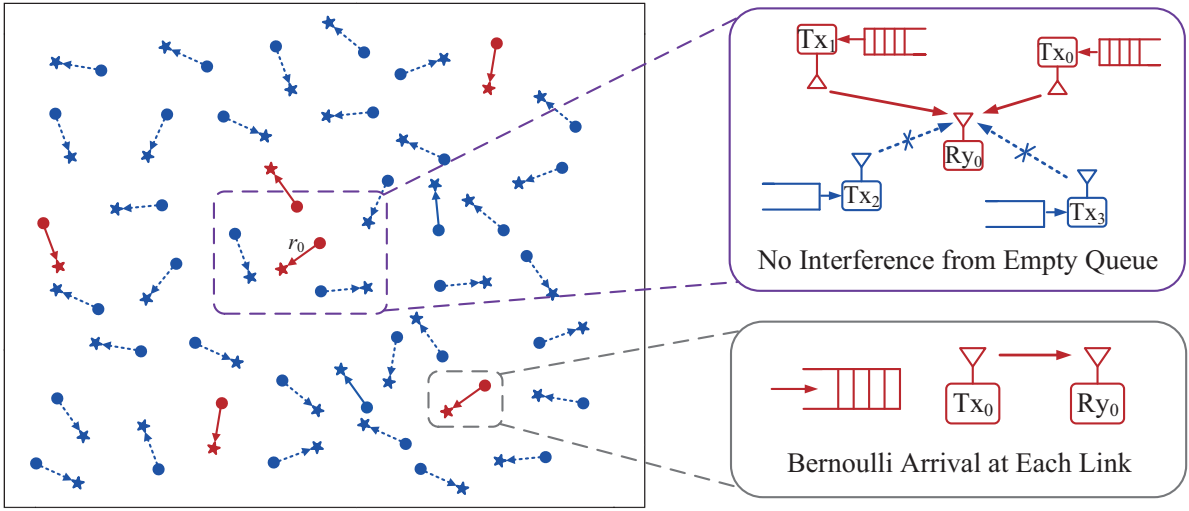}
\caption{The Poisson bipolar model for massive and sporadic access (MSA). The red solid links denote the active links whose queues are non-empty and also selected to transmit by random access, while the blue dotted links represent the silent links whose queues are empty or silenced by random access. The arrival processes at different transmitters are independent Bernoulli processes.}
\label{fig:sysmodel}
\end{figure*}

As for the temporal model, we consider a discrete-time queueing system, where the time is assumed to be divided into discrete time slots with equal duration. The transmission of each packet occupies exactly one time slot. Each transmitter is equipped with an infinite queue to store the incoming packets. In each time slot, if a queue is non-empty, it attempts to transmit its head-of-line packet with probability $p$. If the transmission attempt is successful, the packet will be removed from the queue. Otherwise, the packet will be put back to the head-of-line of the queue and waits to be retransmitted in the next time slot. In the scenario of MSA, the packets arrive at the transmitters as stochastic processes with very small arrival rates. In this paper, we assume that the packets arrival process at each transmitter is a Bernoulli process with arrival rate $\xi$ $(0\leq \xi\leq 1)$, which is widely used in modeling the discrete-time systems. According to the definition of the Bernoulli process, $\xi$ is also the probability that a packet arrives at a transmitter in each time slot. The arrival processes at different transmitters are assumed to be independent of each other.
To be rigorous, we assume that the early arrival model is used where a potential packet departure occurs at the moment immediately before the time slot boundaries, and a potential arrival occurs at the moment immediately after the time slot boundaries (see Figure \ref{fig:earlyarrival}).
If the time axis is marked by $0,1,...,t,...$, a potential departure occurs in the interval $(t^-,t)$, while a potential arrival occurs in the interval $(t,t^+)$.

\begin{figure}
\centering
\includegraphics[width=0.5\textwidth]{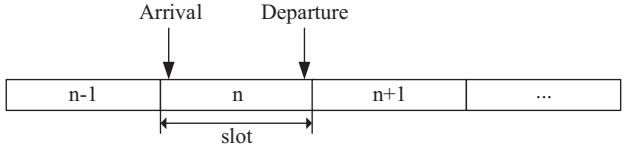}
\caption{Early arrival model where a potential packet departure occurs at the
moment immediately before the time slot boundaries, and a potential arrival occurs at the moment immediately after the time slot boundaries. }
\label{fig:earlyarrival}
\end{figure}

We assume that the network is static, i.e., the locations of all nodes are generated first as a realization of the PPP and then remain unchanged in all the following time slots. This assumption is very realistic since the locations of access points (or base stations, as the receiving nodes) and the devices (as the transmitting nodes)
in most of the MSA wireless networks are fixed after they are deployed.

On the other hand, in general, the nodes may move according to certain mobility model \cite{7399689}.
In the extreme situation, the position of each node can be assumed to change according to a high mobility random walk model \cite[Ch. 1.3]{baccelli2009stochastic}
so that its location in the next time slot can be considered as totally independent from that in the current time slot. We name this scenario with fast node movement as the high mobility scenario. The interference is greatly affected by the type of mobility for the nodes. For example, when the network is static, the nodes that generate interference in next time slots may come from the same set of node locations in the current time slot, which are determined when the nodes are first deployed. However, for the case of high mobility, the locations of nodes that generate interference can change over different time slots.

In this paper, we will demonstrate that a static  MSA network can be approximated by an equivalent high mobility network in terms of their correlation properties.

The transmit power of all nodes is assumed to be the same, which is normalized to be one.
The propagation loss of the electromagnetic wave is assumed to consist of two parts, i.e., the path loss and the fading.
The standard path loss model is used, i.e., the path loss between a transmitter and a receiver with distance $r$ apart is denoted by $l(r)=r^{-\alpha}$, where $\alpha$ is the path loss exponent with $\alpha>2$. The fading model is assumed to be the commonly used Rayleigh block fading with the shadowing being ignored.
The power fading coefficients stay unchanged during each time slot and are spatially and temporally independent under the exponential distribution with unit mean between different time slots.

The normalized power of thermal noise is assumed to be $W$. Note that whether the scenario of MSA belongs to the interference-limited regime or the noise-limited regime cannot be determined since both the node intensity $\lambda$ and the arrival rate $\xi$ may influence the relationship between the interference and the noise, which will be quantitatively evaluated in the following sections.

In each time slot, a transmitter is active only when its queue is non-empty and it is allowed to transmit (with probability $p$). Let $k\in\mathbb{N}^+$ be the index of the time slots and $\Phi_k\in\Phi$ be the set of all active transmitters in the time slot $k$. Note that the set $\Phi_k$ varies with the time since both the statuses of queues and the scheduled results of random access are different in different time slots. Then, the interference at the typical receiver located at the origin $o$ in time slot $k$ is
\begin{equation}
I_k=\sum_{x\in\Phi\backslash\{x_0\}}h_{k,x}|x|^{-\alpha}\mathbf{1}(x\in\Phi_k),
\end{equation}
where $h_{k,x}$ is the fading coefficient between the interfering transmitter $x$ and the typical receiver at the origin.
Since the original MSA network is static, the set of transmitters $\Phi$ is independent of the index of time slot $k$.

When the typical transmitter at $x_0$ attempts to deliver a packet to the typical receiver at the origin, the signal to interference plus noise ratio (SINR) at the typical receiver when it is active in time slot $k$ is
\begin{equation}
\mathrm{SINR}_k=\frac{h_{k,x_0}r_0^{-\alpha}}{\sum_{x\in\Phi\backslash\{x_0\}}h_{k,x}|x|^{-\alpha}\mathbf{1}(x\in\Phi_k)+W}. \label{eqn:SINRk}
\end{equation}

The SINR threshold for successfully delivering a packet is $\theta$, i.e., a transmission attempt of a link is successful only when the SINR of such link is above the threshold $\theta$. Then, the success probability for the typical link when it is active in time slot $k$ is
\begin{equation}
\mathcal{P}_k=\mathbb{P}\{\mathrm{SINR}_k>\theta\}.
\end{equation}

Note that the interval between two adjacent packet arrivals is a geometrically distributed random variable due to the Bernoulli arrival.
Therefore, the queueing process at the typical link is a discrete-time Geo/G/1 queueing system \cite{zhang2001discrete}.
In particular, when $\mathcal{P}_k$ is the same for all $k\in\mathbb{N}^+$, i.e., the success probability for the typical link in all time slots is the same, the service time (in number of time slots) of each packet follows a geometric distribution with the parameter $\mu\triangleq p\mathcal{P}_k, \forall k\in\mathbb{N}^+$. The parameter $\mu$ is also the mean service rate for the queueing system.
In this case, since the service time is also a geometric distributed random variable, the queueing system can be denoted by Geo/Geo/1.
We present the following lemma which gives the mean delay for a Geo/Geo/1 queueing system.

\begin{lem}
\label{lem:meandelay}
For a discrete-time Geo/Geo/1 queueing system with packet arrival rate $\xi$ and service rate $\mu$ ($\xi<\mu$), the mean delay is
\begin{equation}
D=\frac{1-\xi}{\mu-\xi}. \label{eqn:meandelay}
\end{equation}
\end{lem}
\begin{proof}
From \cite[Corollary 2]{atencia2004discrete}, the mean delay (including the queueing delay and the service delay) for a Geo/G/1 queueing system with arrival rate $\xi$ is
\begin{equation}
D=\beta_1+\frac{2\overline{\xi}(\beta_1-1)(1-A(\overline{\xi}))+\xi\beta_2}{2(\xi+\overline{\xi}A(\overline{\xi})-\xi/\mu)},
\end{equation}
where $\overline{\xi}=1-\xi$, $\beta_1$ and $\beta_2$ are the first and the second factorial moments of the service time, $A(\cdot)$ is the generating function of the successive interretrial times.

Since we consider a standard Geo/G/1 queueing system without retrial, i.e., the packet at the head of the queue immediately commences its service whenever the server is idle, the successive interretrial time is always zero, i.e., $A(x)=1, \forall x$. Thus, the mean delay is reduced to
\begin{equation}
D=\beta_1+\frac{\xi\beta_2}{2(1-\xi/\mu)}. \label{eqn:meandelayproof1}
\end{equation}
Note that the service time is a geometric distributed random variable with mean $1/\mu$ for the Geo/Geo/1 queueing system. The first and second raw moments of the service time are $1/\mu$ and $2/\mu^2-1/\mu$ respectively. Then, the first and the second factorial moments of the service time are $\beta_1=1/\mu$ and $\beta_2=2/\mu^2-2/\mu$. Plugging these results into (\ref{eqn:meandelayproof1}), we obtain the lemma.
\end{proof}


\section{Mathematical Description of MSA}
\label{sec:mathMSA}
To explore the performance of MSA, we first give the explicit mathematical definition for the scenario of MSA. The two main parameters related to the MSA are the density of nodes $\lambda$ and the arrival rate $\xi$.
A simple way to define the scenario of MSA is to put absolute limitations on the density $\lambda$ and the arrival rate $\xi$ which are independent of channel parameters. However, this definition of completely ignoring the channel parameters may not be suitable, for example, even the densely deployed nodes may demonstrate the properties of sparse deployment for large path loss exponent $\alpha$. Therefore, in order to characterize the property of ``massive'' and ``sporadic'', we propose the following definitions which use the minimum tolerable success probability and the maximum tolerable delay to determine the ranges of $\lambda$ and $\xi$.
\begin{defn}
\label{def:massive}
A wireless network is defined as ``massive'' if and only if the density of devices $\lambda\in R^+$ satisfies
\begin{equation}
\lambda\geq\lambda_0\triangleq-\frac{1}{C_0}(\theta Wr_0^\alpha+\ln\varepsilon), \label{eqn:massive}
\end{equation}
where $\delta=2/\alpha$, $C_0=\pi \theta^\delta r_0^2\Gamma(1+\delta)\Gamma(1-\delta)$, $\Gamma(\cdot)$ is the gamma function, and $\lambda_0$ denotes the minimum density of transmitters that guarantees the success probability of the typical link being smaller than a predefined threshold $\varepsilon$ ($0<\varepsilon<e^{-\theta Wr^\alpha}$) when $\xi\rightarrow1$ and $p=1$, i.e., the backlogged case with all transmitters being active.
\end{defn}
\begin{rem}
The reason for using (\ref{eqn:massive}) to distinguish whether the number of devices is massive or not is that the propagation features (such as the path loss exponent) may also make sense in deciding whether a scenario belongs to ``massive'' or not.
The assumption with $\xi\rightarrow1$ and $p=1$ is the worst case that the interference is the largest, i.e., all transmitters are active and cause interference to the typical link.
Note that the expression for the success probability in a Poisson network is well-known, which is \cite[Equation (5.14)]{haenggi2012stochastic}
\begin{equation}
\mathcal{P}_{\rm Poisson}=\exp(-\lambda C_0-\theta Wr_0^\alpha). \label{eqn:succPoisson}
\end{equation}
Letting $\mathcal{P}_{\rm Poisson}\leq\varepsilon$, we get the inequality in Definition \ref{def:massive}.
\end{rem}

Note that when ignoring the interference and only considering the effect of the thermal noise, the success probability for the typical link in any time slot is $e^{-\theta Wr_0^\alpha}$. In order to make the limitation on the success probability meaningful, the threshold $\varepsilon$ for success probability should satisfy the following inequality
\begin{equation}
0<\varepsilon<e^{-\theta Wr_0^\alpha}.
\label{eqn:limit}
\end{equation}
Otherwise, if $\varepsilon\geq e^{-\theta Wr_0^\alpha}$, the inequality $\mathcal{P}_{\rm Poisson}\leq\varepsilon$ always holds since $\mathcal{P}_{\rm Poisson}\leq e^{-\theta Wr_0^\alpha}$.

\begin{defn}
\label{def:sporadic}
A wireless network is defined as ``sporadic'' if and only if the arrival rate of packets $\xi\in[0,1]$ satisfies
\begin{equation}
\xi\leq\xi_0\triangleq\frac{\beta e^{-\theta Wr_0^\alpha}-1}{\beta-1}, \label{eqn:sporadic}
\end{equation}
where $\xi_0$ denotes the maximum arrival rate that guarantees the mean delay of the typical link being smaller than a predefined threshold $\beta$ ($\beta>e^{\theta Wr_0^\alpha}$) when $\lambda\rightarrow0$ and $p=1$, i.e., the case without interference and random access.
\end{defn}
\begin{rem}
The assumption $\lambda\rightarrow0$ corresponds to the best case where the interference is ignored.
When the interference is ignored, the success probability for the typical link in any time slot is the same, which is $e^{-\theta Wr_0^\alpha}$.
Then, when $p=1$, the queueing process at the typical link is a Geo/Geo/1 queueing system with the service rate $e^{-\theta Wr_0^\alpha}$ packet per time slot.
If the arrival rate $\xi\geq e^{-\theta Wr_0^\alpha}$, the queue becomes unstable, and the mean delay $D$ will be infinite (i.e., $D=\infty$).
If the arrival rate $\xi$ satisfies $\xi<e^{-\theta Wr_0^\alpha}$, according to the equation (\ref{eqn:meandelay}) in Lemma \ref{lem:meandelay}, the mean delay for each packet is
\begin{equation}
D=\frac{1-\xi}{e^{-\theta Wr_0^\alpha}-\xi}. \label{eqn:rem2delay}
\end{equation}
Letting the mean delay be smaller than the predefined threshold $\beta$, i.e., $D\leq\beta$, we get the inequality in Definition \ref{def:sporadic}.
\end{rem}

Note that when the arrival rate $\xi$ approaches $0$, the mean delay given by (\ref{eqn:rem2delay}) approaches $e^{\theta Wr_0^\alpha}$, which is the smallest mean delay that could be achieved when the interference is ignored and $p=1$.
Therefore, in order to make the limitation on the mean delay meaningful, the threshold $\beta$ for the mean delay should satisfy the following inequality
\begin{equation}
\beta>e^{\theta Wr_0^\alpha}.
\end{equation}

With the above definition for ``massive'' and ``sporadic'', we define the scenario of MSA as follows.
\begin{defn}
\label{def:MSA}
A wireless network is defined as ``\emph{massive and sporadic access}'' if and only if the density of devices satisfies $\lambda\geq\lambda_0$, and the arrival rate of packets satisfies $\xi\leq\xi_0$, where $\lambda_0$ and $\xi_0$ are the critical values given by (\ref{eqn:massive}) and (\ref{eqn:sporadic}).
\end{defn}

\begin{figure}
\centering
\includegraphics[width=0.45\textwidth]{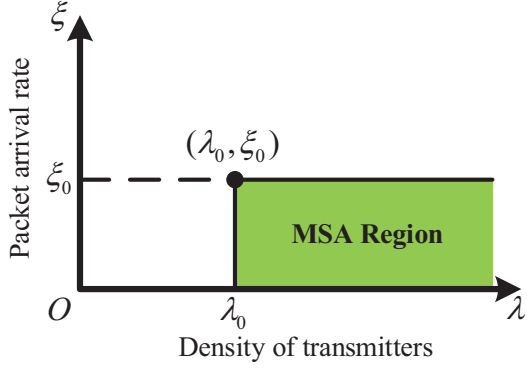}
\caption{Illustration of the region for massive and sporadic access.}
\label{fig:MSAregion}
\end{figure}

In particular, we define the \emph{MSA region} as follows.
\begin{defn}
\label{def:MSAregion}
The MSA region is defined as the range of the two-tuple $(\lambda, \beta)\in R^+\times[0,1]$ within which the corresponding wireless network will be MSA (see Figure \ref{fig:MSAregion}).
\end{defn}

\subsection{Effect of Parameters on MSA Region}

\begin{figure}
  \centering
  \subfigure[Effect of path loss exponent on MSA region.]{
    \label{fig:MSAregion_vary:a} 
    \includegraphics[width=0.45\textwidth]{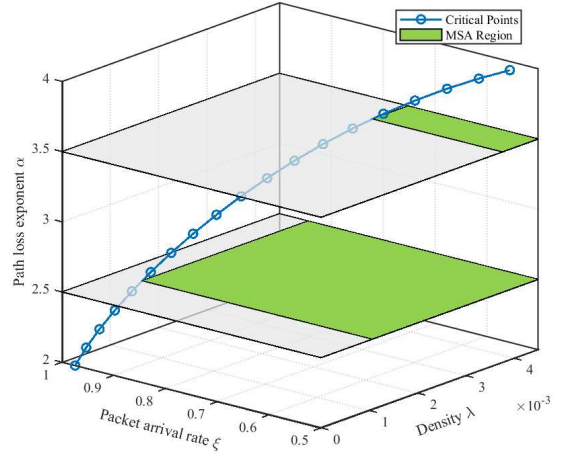}}
  \subfigure[Effect of thermal noise on MSA region.]{
    \label{fig:MSAregion_vary:b} 
    \includegraphics[width=0.45\textwidth]{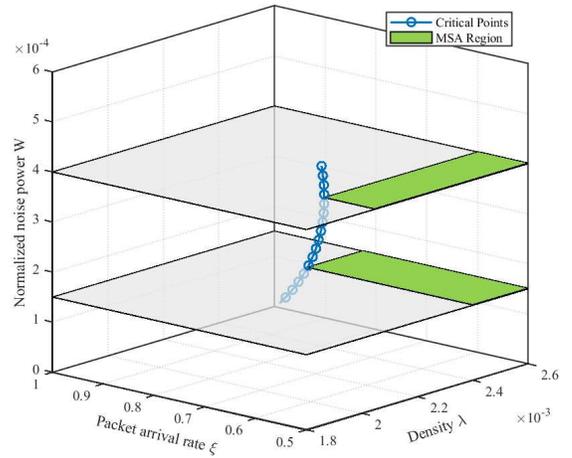}}
  \caption{Illustration of the region for massive and sporadic access. The distance between each transmitter and the associated receiver is $r_0=5$, and the SINR threshold is $\theta=10{\rm dB}$. The threshold for success probability when $\xi\rightarrow1$ is $\varepsilon=0.1$, and the threshold for mean delay when $\lambda\rightarrow0$ is $\beta=50$. The path loss exponent $\alpha$ increases from $2$ to $4$ when fixing the normalized noise $W=10^{-4}$, while $W$ increases from $10^{-6}$ to $10^{-3.3}$ when fixing $\alpha=3$.}
  \label{fig:MSAregion_vary} 
\end{figure}


The shape of MSA region is affected by different network parameters, such as the path loss exponent and the thermal noise. In order to explore the effect of these parameters, we plot the change of MSA region when increasing the path loss exponent $\alpha$ and the normalized thermal noise $W$ in Figure \ref{fig:MSAregion_vary}.
We observe from Figure \ref{fig:MSAregion_vary:a} that the MSA region becomes narrow when increasing $\alpha$. This is because the path loss is enlarged due to the increment of $\alpha$, in which case the network should be denser so that the condition for ``massive'' can be achieved, and the arrival rate of packets should be smaller so that the condition for ``sporadic'' can be achieved. In particular, the density of the critical point goes to zero when $\alpha$ approaches to $2$, which is also observed through the equation (\ref{eqn:massive}), indicating that any density of transmitters could be considered as ``massive'' in the free space propagation model with $\alpha=2$. Figure \ref{fig:MSAregion_vary:b} reveals that the change of $W$ has a great influence on the critical arrival rate $\xi_0$ but less effect on the critical density $\lambda_0$. When $W$ approaches to zero, the critical arrival rate $\xi_0$ goes to $1$, illustrating that any arrival rate will be considered as ``sporadic'' for the case where the thermal noise is ignored.

\subsection{Interference-limited and Noise-limited}
Due to the spectral scarcity, most wireless networks are designed to be interference-limited, i.e. the interference rather than the thermal noise dominates the network performance. However, in the case of MSA, even if the potential transmitters are ultra dense, a wireless network may still be noise-limited since the arrival rate of packets is extremely small resulting in minor interference. In this subsection, we discuss the classification of the scenario of MSA to identify whether it belongs to interference-limited or noise-limited.
\subsubsection{Interference-limited}
In order to quantitatively define the interference-limited regime, we consider a simplified system in which a packet failed for transmitting will be discarded instead of being retransmitted. Therefore, the interference in the simplified system will always be a lower bound for that in the original system. We define a scenario as interference-limited if the ratio between the success probability ignoring the noise and the success probability ignoring the interference is less than a small threshold $\eta$ in the corresponding simplified system. Note that in the simplified system, the probability that a link being active in a time slot equals to the product of the probability that there is a packet arriving in the said time slot $\xi$ and the transmit probability $p$, which is $p\xi$ according to the definition of the Bernoulli arrival process. The success probability for an active link when ignoring the noise and considering only the effect of the interference is
\begin{equation}
\mathcal{P}_{\rm inter}=\exp(-\xi p\lambda C_0). \label{eqn:P_inter1}
\end{equation}
The success probability for an active link when ignoring the interference and considering only the effect of the noise is
\begin{equation}
\mathcal{P}_{\rm noise}=\exp(-\theta Wr_0^\alpha).\label{eqn:P_noise1}
\end{equation}
Combining (\ref{eqn:P_inter1}) and (\ref{eqn:P_noise1}) with the condition for interference-limitation $\mathcal{P}_{\rm inter}/\mathcal{P}_{\rm noise}\leq\eta$, we get
\begin{equation}
\xi\lambda\geq\frac{1}{pC_0}(\theta Wr_0^\alpha-\ln\eta). \label{eqn:boundary_inter}
\end{equation}
The above inequality gives the limitations on $\xi$ and $\lambda$ that makes a scenario interference-limited.

\subsubsection{Noise-limited}
For the noise-limited regime, the success probability for delivering a packet at each time slot is approximated the same, which is $p\exp(-\theta Wr_0^\alpha)$. Then, the queueing process at each link can be considered as a Geo/Geo/1 queueing system with arrival rate $\xi$ and service rate $p\exp(-\theta Wr_0^\alpha)$.
The probability of a queue being non-empty equals to the utilization of the queueing system $\xi\exp(\theta Wr_0^\alpha)/p$. The active transmitters constitute an independent thinning version of the original PPP $\Phi$ with thinning probability $\xi\exp(\theta Wr_0^\alpha)$.
Therefore, the success probability when ignoring the noise and considering only the effect of interference is
\begin{equation}
\mathcal{P}_{\rm inter}=\exp(-\xi \lambda C_0\exp(\theta Wr_0^\alpha)). \label{eqn:P_inter2}
\end{equation}
The success probability when ignoring the interference and
considering only the effect of thermal noise is
\begin{equation}
\mathcal{P}_{\rm noise}=\exp(-\theta Wr_0^\alpha).\label{eqn:P_noise2}
\end{equation}

Similar to the case of interference-limitation, we define a scenario as noise-limited if the ratio between the success probability ignoring the interference and
the success probability ignoring the noise is less than
a small threshold $\eta$.
Combining (\ref{eqn:P_inter2}) and (\ref{eqn:P_noise2}) with the condition for noise-limitation $\mathcal{P}_{\rm noise}/\mathcal{P}_{\rm inter}\leq\eta$, we get
\begin{equation}
\xi\lambda\leq\frac{\theta Wr_0^\alpha+\ln\eta}{C_0\exp(\theta Wr_0^\alpha)}. \label{eqn:boundary_noise}
\end{equation}
The above inequality gives the limitations on $\xi$ and $\lambda$ that makes a scenario noise-limited.

\begin{figure}
\centering
\includegraphics[width=0.5\textwidth]{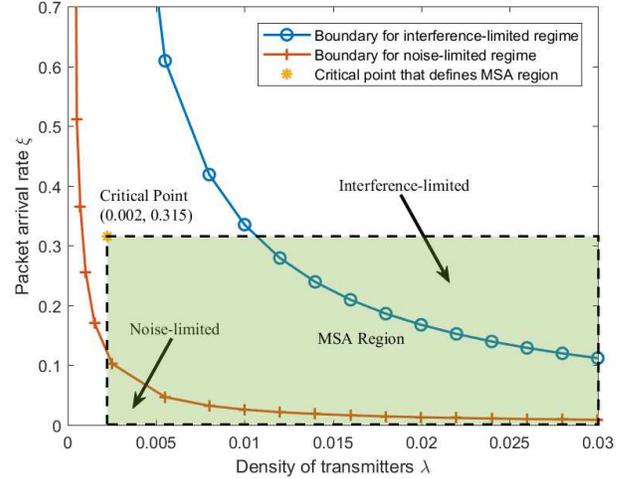}
\caption{Illustration of interference-limited regime and noise-limited regime for the scenario of massive and sporadic access. The distance between each transmitter and the associated receiver is $r_0=5$, the SINR threshold is $\theta=10{\rm dB}$, and the transmit probability is $p=1$. The threshold for success probability when $\xi\rightarrow1$ is $\varepsilon=0.1$, and the threshold for mean delay when $\lambda\rightarrow0$ is $\beta=50$. The path loss exponent is $\alpha=3.5$, the normalized noise is $W=10^{-3.4}$, and the threshold to distinguish different regimes is $\eta=0.5$.}
\label{fig:diffregimes}
\end{figure}

Figure \ref{fig:diffregimes} shows the value of $(\lambda, \xi)$ that belongs to the interference-limited regime or the noise-limited regime in the scenario of MSA. From Figure \ref{fig:diffregimes}, we observe that a wireless network can still be noise-limited as long as the arrival rate is small even if the deployed nodes are highly dense. Figure \ref{fig:diffregimes} also reveals the design insight that there exists a minimum value for the density $\lambda$ to make a network interference-limited in the scenario of MSA. In other words, if the density of the deployed nodes is less than certain value in the scenario of MSA, a wireless network will never be interference-limited for all arrival rate.

\begin{figure}
\centering
\includegraphics[width=0.5\textwidth]{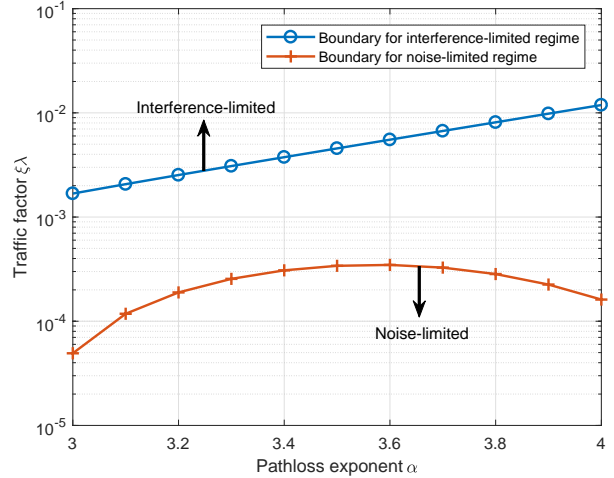}
\caption{Boundaries of the traffic factor $\xi\lambda$ for the interference-limited regime and the noise-limited regime in the scenario of massive and sporadic access. The distance between each transmitter and the associated receiver is $r_0=5$, the SINR threshold is $\theta=10{\rm dB}$, and the transmit probability is $p=1$. The normalized noise is $W=10^{-3.2}$, and the threshold to distinguish different regimes is $\eta=0.5$.}
\label{fig:boundary}
\end{figure}

In general, the product $\xi\lambda$ determines whether a wireless network works in the interference-limited regime or the noise-limited regime. Therefore, we define the product of the arrival rate $\xi$ and the density of transmitters $\lambda$ as the \emph{traffic factor}. Intuitively, the traffic factor describes the intensity of the traffic in a wireless network. In Figure \ref{fig:boundary}, we plot the boundaries of the traffic factor $\xi\lambda$ as functions of the pathloss exponent $\alpha$ for the interference-limited regime given by the inequality (\ref{eqn:boundary_inter}) and for the noise-limited regime given by the inequality (\ref{eqn:boundary_noise}) in the scenario of MSA. Since the Y axis is logarithmic, we observe that the boundary of the traffic factor $\xi\lambda$ for the interference-limited regime increases almost exponentially as the pathloss exponent $\alpha$ grows. We also observe that as $\alpha$ increases, the curve for the boundary of the noise-limited regime first grows and then goes down. This can be intuitively interpreted as that when $\alpha$ starts to increase, the interference decreases rapidly, and small traffic factor $\xi\lambda$ can make the network work in the noise-limited regime. However, if $\alpha$ continues to grow, the success probability of the desired link decreases due to the deterioration of the desired signal, resulting in more retransmissions, which could also be inferred from the equation (\ref{eqn:P_inter2}). In this case, the traffic in the network becomes heavier due to the retransmitted packets, leading to increased interference and more stringent conditions on the traffic fact $\xi\lambda$ to make the network noise-limited.

\section{Temporal Correlations}
\label{sec:equivalence}
In this section, before introducing the approach of high-mobility approximation where there is no temporal correlation between different time slots, we first evaluate the temporal correlations of interference and successful transmission in the original static network. By proving that these correlations are negligible in the scenario of MSA, we demonstrate the reasonability of the high-mobility approximation.

Since the queueing process at each individual transmitter in a wireless network depends the statuses (empty or non-empty) of queues at all transmitters, the queueing processes at different transmitters are highly coupled with each other, leading to the interacting queues problem which is rather difficult to cope with. Existing works have only derived the sufficient conditions and necessary conditions for the stability \cite{7486114}, or the upper and lower bounds for certain performance metrics \cite{7886285}. However, in the case of the MSA, the set of active transmitters that cause interference to the network changes dramatically over the time slots. Moreover, in the scenario of ``extremely massive'' and ``extremely sporadic'', the sets of active transmitters in different time slots may not intersect. Then, the interference in different times slots will be independent, resulting in the decoupling of the queueing processes for different queues in the wireless networks.

Intuitively, the analysis of a wireless network with large number of coupled queues for the scenario of MSA could be approximately equivalent to the analysis of a high mobility case where the interacting queues are decoupled, thereby greatly reducing the analytical complexity. However, the accuracy of such intuition requires to be quantitatively analyzed in order to use this equivalence to simplify the analysis of the MSA. Therefore, in this section, we explore the difference between the original static scenario of MSA and the high mobility case from the view of temporal correlations of interference and successful transmission. Note that in the high mobility network, both the interference and the successful transmission are independent between different time slots, i.e. the temporal correlations of both the interference and the successful transmission tend to zero.

To facilitate the analysis, we consider a backlogged version of the original network, in which the queues at all transmitters are backlogged and will never be empty.
In the backlogged network, each transmitter will be active independently with the same probability $p$.
Note that the interference is temporally correlated since a subset from the same set of potential transmitters are active in different time slots.
The locations of these potential transmitters are randomly deployed first as a realization of the PPP and then keep unchanged in all the following time
slots, which could be considered as the ``common randomness''.

\subsection{Temporal Correlation of Interference}
The temporal (Pearson's) correlation coefficient of the interference in the $i$-th time slot $I_i$ and that in the $j$-th time slot $I_j$ is defined as
\begin{equation}
\rho(I_i,I_j)\triangleq\frac{\mathrm{cov}(I_i,I_j)}{\sigma_{I_i}\sigma_{I_j}},
\end{equation}
where $\mathrm{cov}(I_i,I_j)=\mathbb{E}[I_iI_j]-\mathbb{E}[I_i]\mathbb{E}[I_j]$ is the covariance between $I_i$ and $I_j$, and $\sigma_{I_i}=\sigma_{I_j}=\sqrt{\mathbb{E}[I_i^2]-(\mathbb{E}[I_i])^2}$ is the standard deviation of $I_i$ and $I_j$.
The work in \cite{net:Ganti09CL} has already obtained the temporal correlation coefficient in remarkably simple form as
\begin{equation}
\rho(I_i,I_j)=\frac{p}{\mathbb{E}[h_{k,x}^2]},
\end{equation}
where $h_{k,x}$ is the power fading coefficient between an interfering
transmitter $x$ and the typical receiver at the origin in any time slot $k$.
Due to the Rayleigh fading assumption, the power fading coefficient $h_{k,x}$ is exponentially distributed with unit mean. Then, the temporal correlation coefficient becomes
\begin{equation}
\rho(I_i,I_j)=\frac{p}{2}. \label{eqn:rhoIij}
\end{equation}
The above equation indicates that the correlation coefficient of the interference between two different time slots grows linearly with the active probability $p$. Note that the correlation coefficient is not related to the density of the nodes $\lambda$. This observation indicates that in the scenario of MSA, the degree of temporally linear correlation of the interference will not be affected by increasing the density of links (i.e., more massive), while it decreases linearly as the arrival rate decreases (i.e., more sporadic).

\subsection{Temporal Correlation of Successful Transmission}
In practical system, we care more about the outcome of the transmission attempt rather than the interference level. Note that the success probability relies on the SINR, the denominator of which is determined by the interference and the thermal noise. Therefore, the temporal correlation of the interference induces the temporal correlation of the successful transmissions. Let $S_k=\mathbf{1}(\mathrm{SINR}_k>\theta)$ be the indicator that the SINR at the typical receiver in time slot $k$ is above the threshold $\theta$, i.e., $S_k$ is the indicator for successful transmission in time slot $k$.
The temporal correlation coefficient between $S_i$ and $S_j$, $i\neq j$ is
\begin{equation}
\rho(S_i,S_j)=\frac{\mathbb{E}[S_iS_j]-\mathbb{E}[S_i]\mathbb{E}[S_j]}{{\mathbb{E}[S_i^2]-(\mathbb{E}[S_i])^2}}. \label{eqn:rhoSiSj}
\end{equation}
Note that $\mathbb{E}[S_k^2]=\mathbb{E}[S_k]$ and $\mathbb{E}[S_k]=\mathcal{P}_k$, $\forall k\in\mathbb{N}^+$, we have
\begin{equation}
\rho(S_i,S_j)=\frac{\mathbb{P}\{\mathrm{SINR}_i>\theta,\mathrm{SINR}_j>\theta\}-\mathcal{P}_i^2}{{\mathcal{P}_i-\mathcal{P}_i^2}}.
\end{equation}
Due to the backlogged assumption and the independent thinning, the success probability $\mathcal{P}_i$ is given by the standard form for a Poisson network similar to the equation (\ref{eqn:succPoisson}) as follows
\begin{equation}
\mathcal{P}_i=\exp(- p\lambda C_0-\theta Wr_0^\alpha). \label{eqn:succBacklogged}
\end{equation}
Using the formula (\ref{eqn:SINRk}) and the exponential distribution property of $h_{k,x}$, we obtain the joint success probability as
\begin{eqnarray}
\mathbb{E}[S_iS_j]&\!\!=\!\!&\mathbb{P}\{\mathrm{SINR}_i>\theta,\mathrm{SINR}_j>\theta\} \nonumber\\
&\!\!=\!\!&\mathbb{E}\big\{e^{-\theta (I_i+W) r_0^\alpha}e^{-\theta (I_j+W) r_0^\alpha}\big\} \nonumber\\
&\!\!=\!\!&e^{-2\theta W r_0^\alpha}\mathbb{E}\prod_{x\in\Phi\setminus\{x_0\}}e^{-2\theta r_0^\alpha h_{k,x}|x|^{-\alpha}\mathbf{1}(x\in\Phi_k)} \nonumber\\
&\!\!=\!\!&e^{-2\theta W r_0^\alpha}\mathbb{E}\prod_{x\in\Phi\setminus\{x_0\}}
\Big(1-p+\frac{p}{1+2\theta r_0^\alpha|x|^{-\alpha}}\Big)
 \nonumber\\
 &\!\!\overset{(a)}{=}\!\!&\exp\Big(-2\theta W r_0^\alpha-2\pi p\lambda\int_0^\infty\frac{2\theta r_0^\alpha r\mathrm{d}r}{r^\alpha+2\theta r_0^\alpha}\Big)\nonumber\\
 &\!\!=\!\!&\exp(-2\theta W r_0^\alpha- p\lambda 2^\delta C_0). \label{eqn:jointsucc}
\end{eqnarray}
Plugging the above joint success probability into the equation (\ref{eqn:rhoSiSj}), we obtain the temporal correlation coefficient between $S_i$ and $S_j$ as
\begin{equation}
\rho(S_i,S_j)=\frac{\exp((2-2^\delta) p\lambda C_0)-1}{\exp(\theta Wr_0^\alpha+ p\lambda C_0)-1}. \label{eqn:rhoSiSj2}
\end{equation}

The temporal correlation coefficient of the successful transmissions in different time slots depends on the product of the transmit probability $p$ and the density of transmitters $\lambda$.
From the equation (\ref{eqn:rhoSiSj2}), we observe that when $p\lambda\rightarrow0$ or $p\lambda\rightarrow+\infty$, the temporal correlation coefficient $\rho(S_i,S_j)\rightarrow0$, indicating that the successful transmissions in different time slots tend to be linearly independent when $p\lambda$ is either very large or very small. This can be interpreted as that when $p\lambda\rightarrow0$, the interference could be ignored, and the thermal noise becomes the dominant factor that affects the successful transmission events. Due to the independence of the noise at different receivers, the successful transmission events at different time slots tend to be completely independent. Meanwhile, when $p\lambda\rightarrow+\infty$, the successful transmission events are affected by a large number of independent channel fading coefficients, leading to the independence of the successful transmission events in different time slots. In particular, we have the following lemma.

\begin{lem}
\label{lem:monotonicity}
The temporal correlation coefficient of the successful transmissions $\rho(S_i,S_j)$ is maximized if and only if $p\lambda={C_0^{-1}}{\ln t_0}$, where $t_0\in(1,+\infty)$ is the solution of the equation
\begin{equation}
(a-1)t_0^a-abt_0^{a-1}+1=0, \label{eqn:monoeqn}
\end{equation}
where $a=2-2^\delta$ and $b=\exp(-\theta Wr_0^\alpha)$.
\end{lem}
\begin{proof}
Letting $a=2-2^\delta$, $b=\exp(-\theta Wr_0^\alpha)$, and
$$f(t)=\frac{t^a-1}{t-b},$$ we have $0<a, b<1$ and $\rho(S_i,S_j)=bf(\exp(C_0p\lambda))$.
Note that the derivative of $f(t)$ is
\begin{equation}
f'(t)=\frac{(a-1)t^a-abt^{a-1}+1}{(t-b)^2}. \label{eqn:fder}
\end{equation}
It can be verified that the numerator $g(t)=(a-1)t^a-abt^{a-1}+1$ is monotone decreasing when $t>1$. Thus, $f'(t)$ is also monotone decreasing when $t>1$, demonstrating that $f(t)$ is a concave function when $t>1$. Since the equality $f(1)=f(+\infty)=0$ holds, there is one and only one maximum value of $f(t)$ for $t\in(1,+\infty)$. Let $t_0=\mathop{\arg\max}\limits_{t\in(1,+\infty)}f(t)$, and $t_0$ should satisfy the following equation
\begin{equation}
(a-1)t_0^a-abt_0^{a-1}+1=0.
\end{equation}
Due to the equation $\rho(S_i,S_j)=bf(\exp(C_0p\lambda))$, we get the result in Lemma \ref{lem:monotonicity}.
\end{proof}

\begin{figure}
\centering
\includegraphics[width=0.5\textwidth]{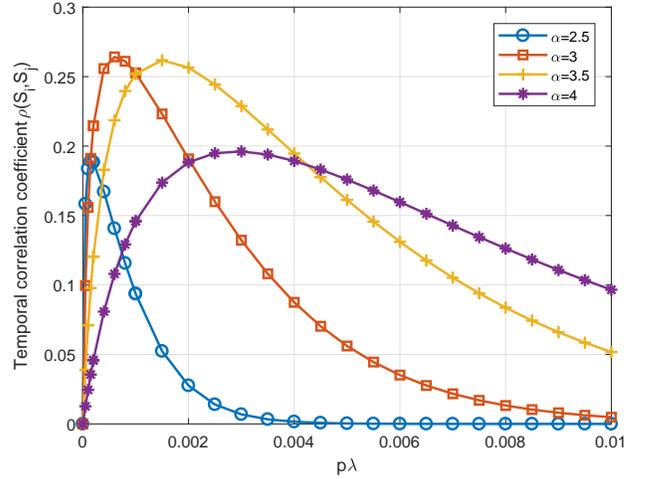}
\caption{Temporal correlation coefficient of the successful transmissions $\rho(S_i,S_j)$ as functions of $p\lambda$ for different path loss exponents. The distance between each transmitter and the associated receiver is $r_0=5$, and the SINR threshold is $\theta=10{\rm dB}$. The normalized noise is $W=10^{-4}$.}
\label{fig:SuccCorr}
\end{figure}

In Figure \ref{fig:SuccCorr}, we plot the temporal correlation coefficient of the successful transmissions $\rho(S_i,S_j)$ as functions of $p\lambda$ for different path loss exponents. Figure \ref{fig:SuccCorr} verifies Lemma \ref{lem:monotonicity} that there is one and only one stationary point that maximizes the temporal correlation coefficient. We also observe that as the path loss exponent $\alpha$ increases from $2.5$ to $4$, the maximum temporal correlation coefficient first grows then decreases. This is because the attenuation of the interference is small for small $\alpha$, in which case the interference comes from plentiful sources with non-negligible interfering signals. Since the interference levels from different sources are independent due to the independent fading, the successful transmission events in different time slots become less correlated.
When $\alpha$ is large, the attenuation of the interference is remarkable, and the thermal noise which is independent at different receivers becomes dominant, leading to the reduction of the temporal correlation.

\begin{figure}
\centering
\includegraphics[width=0.5\textwidth]{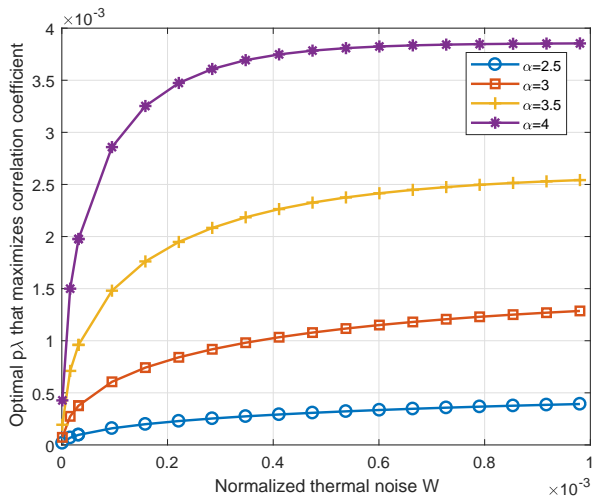}
\caption{The optimal $p\lambda$ that maximizes temporal correlation coefficient $\rho(S_i,S_j)$ as functions of the normalized thermal noise $W$ for different path loss exponents. The distance between each transmitter and the associated receiver is $r_0=5$, and the SINR threshold is $\theta=10{\rm dB}$.}
\label{fig:MaxSuccCorr}
\end{figure}

In Figure \ref{fig:MaxSuccCorr}, we plot the optimal $p\lambda$ that maximizes temporal correlation coefficient $\rho(S_i,S_j)$ as functions of the normalized thermal noise $W$ for different path loss exponents. It is observed that the optimal $p\lambda$ that maximizes temporal correlation increases as the path loss exponent $\alpha$ or the normalized thermal noise $W$ grows. Intuitively, the thermal noise becomes dominant other than the interference when $\alpha$ or $W$ increases, and $p\lambda$ should be enlarged to maintain the same level of correlation. Moreover, Figure \ref{fig:MaxSuccCorr} implies that for large thermal noise $W$, the optimal $p\lambda$ that maximizes the temporal correlation tends to stabilize at the same value. In particular, if $W\rightarrow+\infty$, the equation (\ref{eqn:monoeqn}) becomes $(a-1)t_0^a+1=0$. By solving this equation, we get the optimal $p\lambda$ that maximizes temporal correlation of successful transmission events when $W\rightarrow+\infty$ as
\begin{equation}
p\lambda=-\frac{\ln(1-a)}{aC_0}=-\frac{\ln(2^\delta-1)}{(2-2^\delta)C_0}.
\end{equation}

Through the expression for the temporal correlation coefficient of the interference given by (\ref{eqn:rhoIij}), the transmit probability $p$ should be very small so that the interference correlation among different time slots can be ignored. On the other hand, through the expression for the temporal correlation coefficient of the successful transmission events given by (\ref{eqn:rhoSiSj2}), $p\lambda$ should either be very small or very large so that the successful transmission correlation among different time slots can be ignored. Therefore, the temporal correlation of both the interference and the successful transmission can be ignored either when both $p$ and $\lambda$ are very small (corresponding to small $p$ and small $p\lambda$) or when $p$ is very small and $\lambda$ is very large (corresponding to small $p$ and large $p\lambda$). The former case where both $p$ and $\lambda$ are very small implies the noise-limited regime that the effect of the interference could be ignored, while the latter case where $p$ is very small and $\lambda$ is very large implies the aforementioned MSA scenario. In both cases, the wireless network can be equivalently analyzed by using a corresponding high mobility assumption networks (which has the same correlation property). Noting that the noise-limited regime is relatively much easier to evaluate quantitatively due to the approximation of ignoring the interference, we mainly focus our research on the MSA scenario (i.e. with small $p$ and large $\lambda$ values) throughout this paper.

In this section, we have discussed the temporal correlation coefficients of the interference and the successful transmission events by considering a backlogged network in which the queues at all transmitters will never be empty. The transmission of each link is controlled by the transmit probability $p$. For the original network with complex queue evolutions, these temporal correlation coefficients are extremely complicated to quantify. However, the results obtained in the backlogged network illustrate heuristically the temporal correlation in the original system.

\section{Equivalence by High Mobility}
\label{sec:accuracy}
In a practical wireless network with random packet arrival and mutual interference, the longer-term network performance metrics, such as the delay and the stability, are extremely difficult to quantify due to the interacting queues problem which is notoriously hard to analyze \cite{rao1988stability, ephremides1987delay, 7886285}.
In the previous discussions, fortunately, we have illustrated that a static wireless network in the scenario of MSA could be approximated by an equivalent high mobility case (in terms of their correlation properties), which greatly reduces the complexity of the quantitative analysis and design. In order to assess the accuracy of using the high mobility equivalence, we evaluate the empty probability and the mean delay at the stationary regime for the typical link in the scenario of MSA. Furthermore, we compare these two metrics for the original static network and the equivalent high mobility network.

In the high mobility case, the locations of all transmitters in current time slot can be considered as independent with those in the next time slot. Therefore, the spatial distributions of transmitters in different time slots are independent PPP with the same intensity $\lambda$. Unlike the static network where two queues at two nearby transmitters may affect the packet delivery of each other for long periods of time, the mutual effect between the same batch of queues in the high mobility case appears only in the current time slot. Thus, the interacting queues can be decoupled and considered as independent from each other in the high mobility case. Moreover, in the static network, the stationary distributions of the queue systems at different transmitters are diversified, while in the high mobility case the stationary distribution is the same for all queues.

Let $\zeta$ be the probability that a queue is non-empty at the stationary regime in the high mobility case. Then, all queues in the high mobility case are non-empty with the same probability $\zeta$ at the stationary regime. Since only the transmitters with non-empty queues are active with probability $p$, the locations of all active transmitters interfering with the typical link constitute a PPP with intensity $p\zeta\lambda$ in each time slot.
When the typical transmitter at $x_0$ attempts to deliver a packet at the stationary regime, the success probability in different time slots is the same, which is
\begin{equation}
\mathcal{P}_k=\exp\big(-p\zeta\lambda C_0-\theta Wr_0^\alpha\big), \forall k\in\mathbb{N}^+. \label{eqn:succHighMobility}
\end{equation}

In such a high mobility case, the queueing process at the typical link is a Geo/Geo/1 queueing system since the arrival process is geometric due to the Bernoulli arrival, and the service process is also geometric due to the retransmission mechanism with the same transmit probability $p$ and the same success probability $\mathcal{P}_k$. For the Geo/Geo/1 queueing system at the typical link, the arrival rate is $\xi$ while the service rate is $\mu=p\mathcal{P}_k$. Note that the condition for the queue at the typical link being stable is $\xi<\mu$, which turns into
\begin{equation}
\xi<p\exp\big(-p\zeta\lambda C_0-\theta Wr_0^\alpha\big). \label{eqn:stablecondition}
\end{equation}

The probability that the queue at the typical link is non-empty at the stationary regime can be obtained by evaluating the utilization of the Geo/Geo/1 queueing system, which turns into
\begin{equation}
\mathbb{P}\{\text{Queue\ is\ non-empty}\}=\frac{\xi}{\mu}=\frac{\xi}{p\mathcal{P}_k}. \label{eqn:non_empty}
\end{equation}
Plugging (\ref{eqn:succHighMobility}) into (\ref{eqn:non_empty}), we obtain the probability of the queue being empty, i.e. the \emph{non-empty probability}, as
\begin{equation}
\mathbb{P}\{\text{Queue\ is\ non-empty}\}=\frac{\xi}{p}\exp\big(p\zeta\lambda C_0+\theta Wr_0^\alpha\big). \label{eqn:non_empty2}
\end{equation}
Since the typical link is arbitrarily selected from a wireless network, the non-empty probability of the queue at the typical link equals to the value $\zeta$ assumed above. Therefore, we can get a fixed-point equation as follows.
\begin{equation}
\frac{\xi}{p}\exp\big(p\zeta\lambda C_0+\theta Wr_0^\alpha\big)=\zeta. \label{eqn:fixpoint}
\end{equation}
Letting $\omega=-p\zeta\lambda C_0$, the above fixed-point equation (\ref{eqn:fixpoint}) transforms into
\begin{equation}
\omega\exp(\omega)=-\xi\exp(\theta Wr_0^\alpha)\lambda C_0. \label{eqn:fixpoint2}
\end{equation}

Let $\mathcal{W}(z)$ be the Lambert $\mathcal{W}$ function which solves the equation $\mathcal{W}(z)e^{\mathcal{W}(z)}=z$. The Lambert $\mathcal{W}$ function has two branches, i.e., the principal branch $\mathcal{W}_0(z)$ and the branch $\mathcal{W}_{-1}(z)$.
The result obtained from the branch $\mathcal{W}_{-1}(z)$ is rejected since it leads to a system with a success probability which increases with arrival rate.
Using the principal branch of the Lambert $\mathcal{W}$ function $\mathcal{W}_0(z)$ to solve the above fixed-point equation (\ref{eqn:fixpoint2}), we get the solution $\omega_0$ as follows
\begin{equation}
\omega_0=\mathcal{W}_0(-\xi\lambda C_0\exp(\theta Wr_0^\alpha)).
\end{equation}
Therefore, we get the following lemma
\begin{lem}
\label{lem:zeta0}
In a network with high mobility, the non-empty probability for all queues at the stationary regime is the same, which is
\begin{equation}
\zeta_0=-\frac{1}{p\lambda C_0}\mathcal{W}_0(-\xi\lambda C_0\exp(\theta Wr_0^\alpha)). \label{eqn:zeta0}
\end{equation}
\end{lem}

Plugging (\ref{eqn:zeta0}) into (\ref{eqn:stablecondition}), we get the following theorem.
\begin{thm}
\label{thm:stablecondition}
In a network with high mobility, the stable condition for each queue is
\begin{equation}
\xi<p\exp\big(\mathcal{W}_0(-\xi\lambda C_0\exp(\theta Wr_0^\alpha))\\-\theta Wr_0^\alpha\big). \label{eqn:stablecondition2}
\end{equation}
\end{thm}

Plugging (\ref{eqn:zeta0}) into (\ref{eqn:succHighMobility}), we get the success probability in the following theorem.
\begin{thm}
\label{thm:succHighMobility2}
In a network with high mobility, the success probability for all active transmissions at the stationary regime is the same, which is
\begin{equation}
\mathcal{P}_0=\exp\big(\mathcal{W}_0(-\xi\lambda C_0\exp(\theta Wr_0^\alpha))\\-\theta Wr_0^\alpha\big).
\label{eqn:succHighMobility2}
\end{equation}
\end{thm}

The service rate of a Geo/Geo/1 queueing system at the typical link equals to the success probability. By Lemma \ref{lem:meandelay} and Theorem \ref{thm:succHighMobility2}, we get the mean delay in the following corollary.
\begin{cor}
\label{thm:meandelayhigh}
In a network with high mobility, the mean delay for all queues at the stationary regime is the same, which is
\begin{equation}
D_0=\frac{1-\xi}{p\exp\big(\mathcal{W}_0(-\xi\lambda C_0\exp(\theta Wr_0^\alpha))\\-\theta Wr_0^\alpha\big)-\xi}. \label{eqn:meandelayhigh}
\end{equation}
\end{cor}

According to the Little's Law, the average number of packets in a queue at the stationary regime $L_0$ equals to the product of the arrival rate and the mean delay, i.e., $L_0=\xi D_0$. Thus, we get the average queue length as
\begin{eqnarray}
L_0=\frac{\xi(1-\xi)}{p\exp\big(\mathcal{W}_0(-\xi\lambda C_0\exp(\theta Wr_0^\alpha))-\theta Wr_0^\alpha\big)-\xi}. \label{eqn:avelength}
\end{eqnarray}

\begin{figure}
\centering
\includegraphics[width=0.5\textwidth]{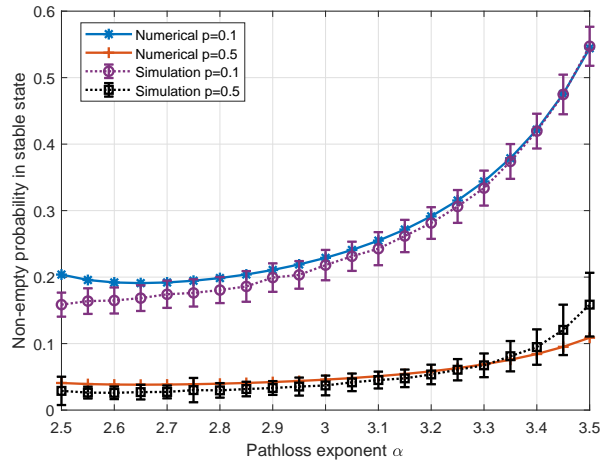}
\caption{Non-empty probability in the stable state as functions of the pathloss exponent $\alpha$. The error bars show the standard deviation.}
\label{fig:NonemptyCom}
\end{figure}

\begin{figure}
\centering
\includegraphics[width=0.5\textwidth]{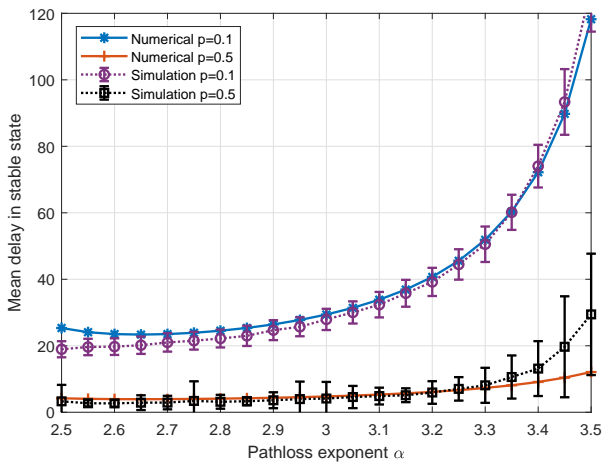}
\caption{Mean delay in stable state as functions of the pathloss exponent $\alpha$. The error bars show the standard deviation.}
\label{fig:MeandelayCom}
\end{figure}

In the following discussion, we analyze the success probability and the mean delay through the numerical evaluation. The distance between each transmitter and the associated receiver is $r_0=5$, and the SINR threshold is $\theta=10{\rm dB}$. The normalized noise is $W=10^{-3.3}$.

In Figure \ref{fig:NonemptyCom} and Figure \ref{fig:MeandelayCom}, we compare the simulation results for the scenario of MSA and the numerical results obtained by the high mobility equivalence to evaluate the accuracy of the proposed equivalent analysis approach. The density of transmitters is $\lambda=0.01$, and arrival rate is $\xi=0.01$. The simulation for the scenario of MSA is conducted in an area of size $240\times240$, with the statistics within the margins less than $20$ are ignored to eliminate the edge effect (i.e., the links at the edge experience less interference due to the finite edge). The network topology is regenerated for $200$ times according to the point process, and for each realization of the point process, the duration of the simulation is $1000$ time slots. Figure \ref{fig:NonemptyCom} and Figure \ref{fig:MeandelayCom} reveal that the proposed approximating approach by the equivalence of high mobility provides a good characterization for both non-empty probability and mean delay in the stable state for MSA.

\begin{figure}
\centering
\includegraphics[width=0.5\textwidth]{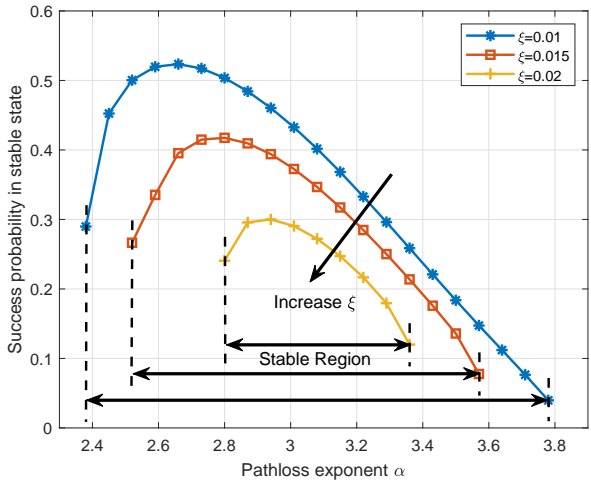}
\caption{Success probability in the stable state as functions of the pathloss exponent $\alpha$. The transmit probability is $p=0.5$.}
\label{fig:SuccProb}
\end{figure}

\begin{figure}
\centering
\includegraphics[width=0.5\textwidth]{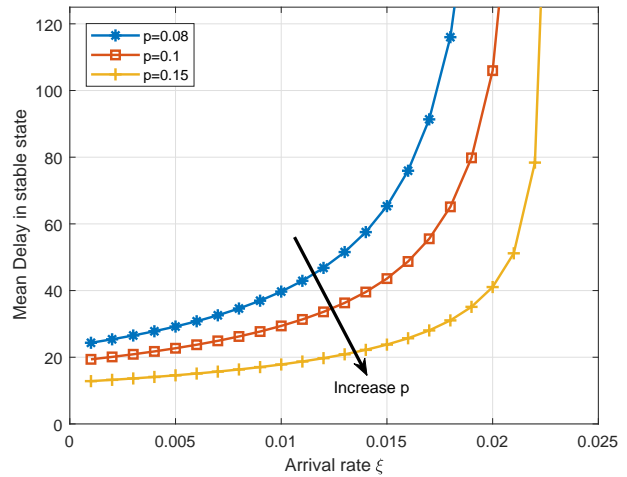}
\caption{Mean delay in stable state as functions of the arrival rate $\xi$. The pathloss exponent is $\alpha=3$.}
\label{fig:Meandelay}
\end{figure}

In Figure \ref{fig:SuccProb}, we plot the success probability in the stable state as functions of the pathloss exponent $\alpha$ for different arrival rate $\xi$ with the high mobility equivalence. We observe that the success probability is limited by large interference for small $\alpha$ and by small desired signal power for large $\alpha$. We also observe that the range of $\alpha$ which makes the network stable become smaller when $\xi$ increases.
In Figure \ref{fig:Meandelay}, we plot the mean delay in the stable state as functions of the arrival rate $\xi$ for different transmit probability $p$ with the high mobility equivalence. Figure \ref{fig:Meandelay} reveals that the mean delay increases rapidly when the arrival rate $\xi$ approaches certain critical value, beyond which the network will be unstable. For the practical system design, Figure \ref{fig:SuccProb} helps to understand to what extent the arrival rate could be increased to guarantee certain success probability target, and Figure \ref{fig:Meandelay} can be used to configure the protocol parameter $p$ to achieve the mean delay target. These plots give examples for how the proposed approaches are used to reveal the relationship between performance metrics and deployment parameters.

Therefore, as long as the conditions for `massive and sporadic' are met, the equivalence of high mobility will provide an accurate approximation for the original static network.
The proposed approaches of high mobility approximation is convenient to be used in the practical system design of massive and sporadic access since the analysis of a wireless network with high mobility is far simpler to handle than a static network within which complicated coupling among the queues persists.
Moreover, the derivations based on high mobility approximation will most of the time lead to closed-form solutions based on which the insights can be directly observed from the closed-form formulas.

\section{Conclusions}
\label{sec:conclusions}
Evaluation of longer-term performance metrics of a wireless network, such as delay and stability, by using stochastic geometry has long been a difficult problem due to the interacting queues problem.
In this paper, we define the scenario of MSA and derive the temporal correlation coefficients for interference and successful transmission events at different time slots to demonstrate that these correlations are negligible in the scenario of MSA. Hence, we propose to use an equivalent high mobility network to evaluate the performance for the scenario of MSA, such as non-empty probability, success probability, average queue length, and mean delay.
Moreover, we compare the non-empty probability and the mean delay for the original static network and the equivalent network of high mobility to demonstrate the accuracy of the proposed approach.

The proposed approach is promising for providing a convenient and universal solution for the design of IoT networks with massive and sporadic access, which is far simpler to handle than a static network where complicated coupling among queues persists. Future works are still needed to explore the accuracy of the proposed approach for the cases where the spatial distribution of nodes exhibits aggregation or repulsion. Meanwhile, the adjustable finite mobility could also be evaluated to provide a more precise approximation.

\bibliographystyle{IEEEtran}
\bibliography{ref}

\begin{IEEEbiography}[{\includegraphics[width=1in,height=1.25in,clip,keepaspectratio]{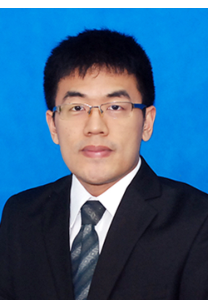}}]{Yi Zhong} (S'12-M'15) received his B.S. and Ph.D. degree in Electronic Engineering from University of Science and Technology of China (USTC) in 2010 and 2015 respectively. From 2015 to 2016, he was a Postdoctoral Research Fellow with the Singapore University of Technology and Design (SUTD) in the Wireless Networks and Decision Systems (WNDS) Group. Now, he is an assistant professor with School of Electronic Information and Communications, Huazhong University of Science and Technology, Wuhan, China. He is an editor of the IEEE Wireless Communications Letters (since 2020), EURASIP on Wireless Communication and Networking (since 2019), Elsevier Physical Communication (since 2019). His main research interests include heterogeneous and femtocell-overlaid cellular networks, wireless ad hoc networks, stochastic geometry and point process theory.
\end{IEEEbiography}

\begin{IEEEbiography}[{\includegraphics[width=1in,height=1.25in,clip,keepaspectratio]{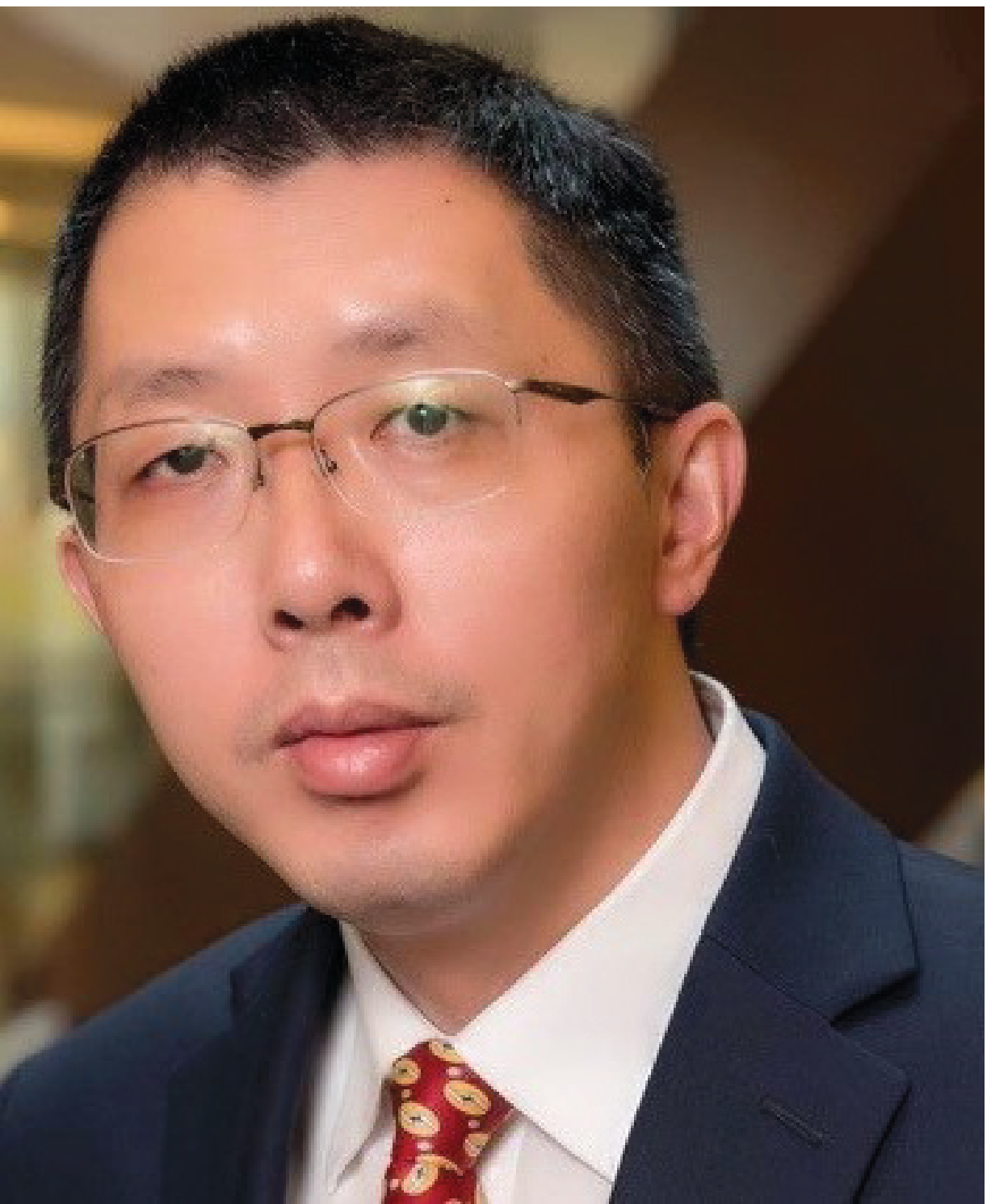}}]{Guoqiang Mao} (S'98-M'02-SM'08-F'18) is a Distinguished Professor at Xidian University. Before that he was with University of Technology Sydney and the University of Sydney. He has published over 200 papers in international conferences and journals, which have been cited more than 9000 times. He is an editor of the IEEE Transactions on Intelligent Transportation Systems (since 2018), IEEE Transactions on Wireless Communications (2014-2019), IEEE Transactions on Vehicular Technology (since 2010) and received ¡°Top Editor¡± award for outstanding contributions to the IEEE Transactions on Vehicular Technology in 2011, 2014 and 2015. He was a co-chair of IEEE Intelligent Transport Systems Society Technical Committee on Communication Networks. He has served as a chair, co-chair and TPC member in a number of international conferences. He is a Fellow of IET. His research interest includes intelligent transport systems, applied graph theory and its applications in telecommunications, Internet of Things, wireless sensor networks, wireless localization techniques and network modeling and performance analysis.
\end{IEEEbiography}

\begin{IEEEbiography}[{\includegraphics[width=1in,height=1.25in,clip,keepaspectratio]{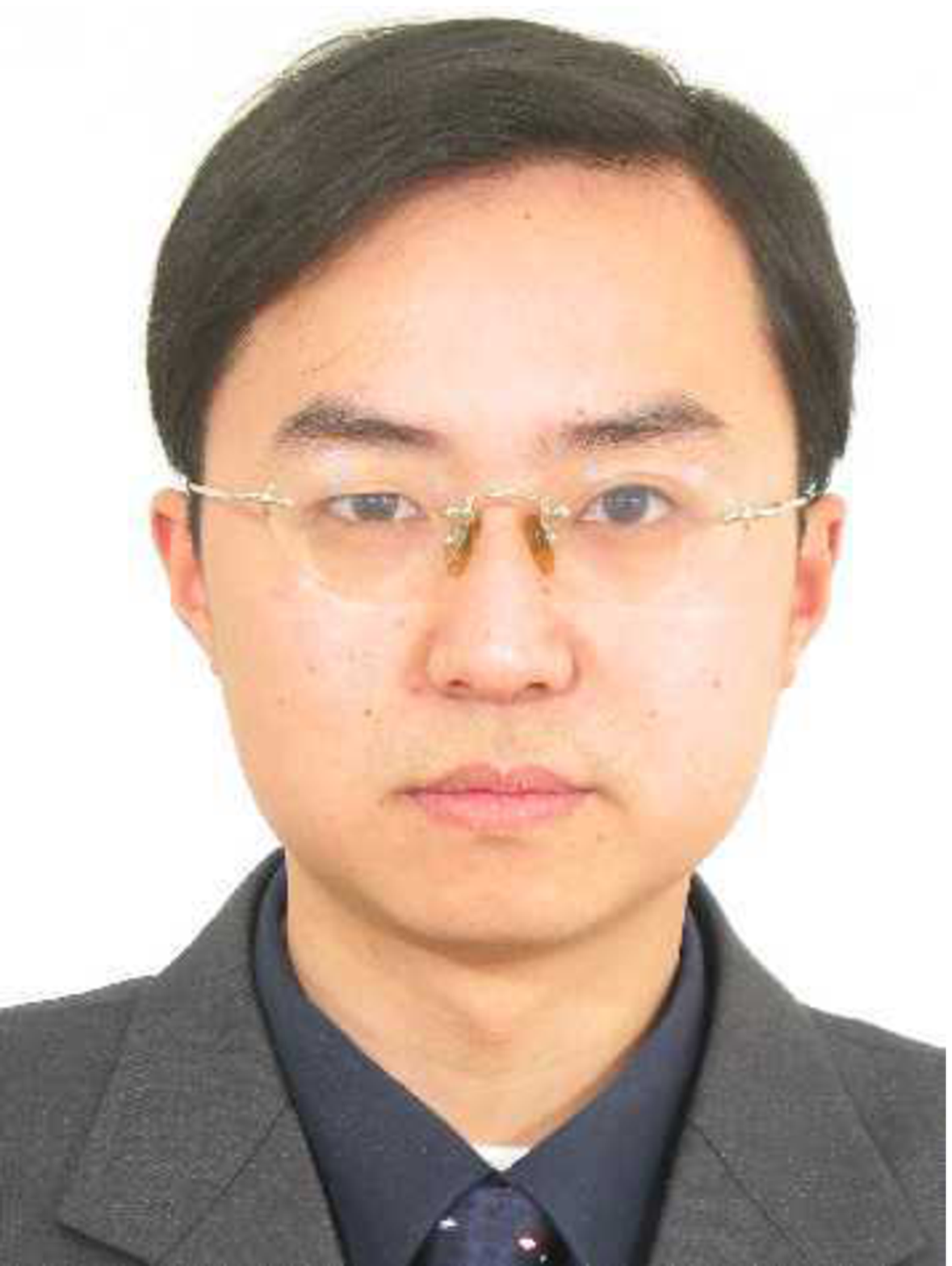}}]{Xiaohu Ge} (M'09-SM'11) received the Ph.D. degree in communication and information engineering from the Huazhong University of Science and Technology (HUST), Wuhan, China, in 2003. He is currently a Full Professor with the School of Electronic Information and Communications, HUST. He is an Adjunct Professor with the Faculty of Engineering and Information Technology, University of Technology Sydney, Ultimo, NSW, Australia. He worked as a Researcher with Ajou University, Suwon, South Korea, and the Politecnico Di Torino, Turin, Italy, from January 2004 to October 2005. He has worked with HUST since November 2005. He has published more than 200 papers in refereed journals and conference proceedings and has been granted about 25 patents in China. His research interests are in the area of mobile communications, traffic modeling in wireless networks, green communications, and interference modeling in wireless communications. Prof. Ge received the Best Paper Awards from IEEE Globecom 2010. He served as the General Chair for the 2015 IEEE International Conference on Green Computing and Communications (IEEE GreenCom 2015). He serves as an Associate Editor for IEEE Wireless Communications, the IEEE Transactions on Vehicular Technology, and IEEE Access.
\end{IEEEbiography}

\begin{IEEEbiography}[{\includegraphics[width=1in,height=1.25in,keepaspectratio]{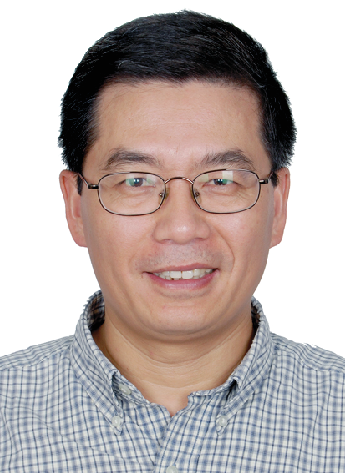}}]
{Fu-Chun Zheng} (M'95-SM'99) obtained the BEng (1985) and MEng (1988) degrees in radio engineering from Harbin Institute of Technology, China, and the PhD degree in Electrical Engineering from the University of Edinburgh, UK, in 1992.

From 1992 to 1995, he was a post-doctoral research associate with the University of Bradford, UK, Between May 1995 and August 2007, he was with Victoria University, Melbourne, Australia, first as a lecturer and then as an associate professor in mobile communications.  He was with the University of Reading, UK, from September 2007 to July 2016 as a Professor (Chair) of Signal Processing. He has also been a distinguished adjunct professor with Southeast University, China, since 2010. Since August 2016, he has been with Harbin Institute of Technology (Shenzhen), China, as a distinguished professor, and the University of York, UK. He has been awarded two UK EPSRC Visiting Fellowships - both hosted by the University of York (UK): first in August 2002 and then again in August 2006. Over the past two decades, Dr Zheng has also carried out many government and industry sponsored research projects - in Australia, the UK, and China. He has been both a short term visiting fellow and a long term visiting research fellow with British Telecom, UK. Dr Zheng¡¯s current research interests include signal processing for communications, multiple antenna systems, green communications, and ultra-dense networks.

He has been an active IEEE member since 1995. He was an editor (2001-2004) of IEEE Transactions on Wireless Communications. In 2006, Dr Zheng served as the general chair of IEEE VTC 2006-S, Melbourne, Australia (www.ieeevtc.org/vtc2006spring) - the first ever VTC held in the southern hemisphere in VTC¡¯s history of six decades. More recently he was the executive TPC Chair for VTC 2016-S, Nanjing, China (the first ever VTC held in mainland China: www.ieeevtc.org/vtc2016spring).
\end{IEEEbiography}

\end{document}